\documentclass[12pt]{article}

\usepackage{psfrag}
\usepackage{graphicx}
\usepackage{amssymb}
\usepackage{times}
\usepackage{mathrsfs}
\usepackage{epsfig}
\usepackage{epsf}

\newtheorem{theorem}{Theorem}
\newtheorem{lemma}{Lemma}
\newtheorem{corollary}{Corollary}

\def\boxit#1{\vbox{\hrule\hbox{\vrule\kern3pt
  \vbox{\kern3pt#1\kern3pt}\kern3pt\vrule}\hrule}}
\def\Box{\boxit{\null}}
\newenvironment{proof}{\trivlist\item[]\emph{Proof}:}%
{\unskip\nobreak\hskip 1em plus 1fil\nobreak$\Box$
\parfillskip=0pt%
\endtrivlist}

\makeatletter

\usepackage[noend]{algorithmic}

\newcommand{\PRENDE}{\mathinner{\leftarrow}}

\newcommand{\DDD}{\mathscr{D}}
\newcommand{\SSS}{\mathscr{S}}
\newcommand{\TTT}{\mathscr{T}}
\newcommand{\pp}{\mathinner{\ldotp\ldotp}} 
\newcommand{\dslcp}{\mathit{DS}_{\!\mathit{lcp}}}
\newcommand{\lcp}{\mathit{lcp}}
\newcommand{\rmq}{\mathit{rmq}}

\newcommand{\lca}{\mathit{lca}}
\newcommand{\slink}{\mathit{sl}}
\newcommand{\bestfriend}{\mathit{best.friend}}
\newcommand{\pathstring}{\mathscr{P}}
\newcommand{\bestlcp}{\mathit{best.lcp}}
\newcommand{\ancestors}{\mathscr{A}}
\newcommand{\sibnum}{\mathit{sib}}
\newcommand{\lcplist}{\mathscr{E}}
\newcommand{\clones}{\mathscr{C}}

\begin{document}

\title{Managing Unbounded-Length Keys in Comparison-Driven Data Structures with Applications to On-Line Indexing\footnote{Parts of this paper appeared as extended abstracts in~\cite{AKLL05,FG04}.}}
\small
\author{\textit{Amihood Amir}\\[\smallskipamount]
  Department of Computer Science\\
  Bar-Ilan University, Israel\\
  and\\
  Department of Computer Science\\
  John Hopkins University, Baltimore MD\\
  \and
  \textit{Gianni Franceschini}\\[\smallskipamount]
  Dipartimento di Informatica\\
  Universit\`a di Pisa, Italy\\
  \and
  \textit{Roberto Grossi}\\[\smallskipamount]
  Dipartimento di Informatica\\
  Universit\`a di Pisa, Italy\\
  \and
  \textit{Tsvi Kopelowitz}\\[\smallskipamount]
  Department of Computer Science\\
  Bar-Ilan University, Israel\\
  \and
  \textit{Moshe Lewenstein}\\[\smallskipamount]
  Department of Computer Science\\
  Bar-Ilan University, Israel\\
  \and
  \textit{Noa Lewenstein}\\[\smallskipamount]
  Department of Computer Science\\
  Netanya College, Israel\\
} \date{}

\maketitle

\thispagestyle{plain}

\newpage

\begin{abstract}
  This paper presents a general technique for optimally transforming
  any dynamic data structure that operates on atomic and indivisible
  keys by constant-time comparisons, into a data structure that
  handles unbounded-length keys whose comparison cost is not a
  constant.  Examples of these keys are strings, multi-dimensional
  points, multiple-precision numbers, multi-key data (e.g.~records),
  XML paths, URL addresses, etc.  The technique is more general than
  what has been done in previous work as no particular exploitation of
  the underlying structure of is required.  The only requirement is
  that the insertion of a key must identify its predecessor or its
  successor.

  Using the proposed technique, online suffix tree construction can be
  done in worst case time $O(\log n)$ per input symbol (as opposed to
  amortized $O(\log n)$ time per symbol, achieved by previously known
  algorithms). To our knowledge, our algorithm is the first that
  achieves $O(\log n)$ worst case time per input symbol. Searching for
  a pattern of length $m$ in the resulting suffix tree takes $O(\min(m
  \log |\Sigma|, m + \log n) + tocc)$ time, where $tocc$ is the number
  of occurrences of the pattern. The paper also describes more
  applications and show how to obtain alternative methods for dealing
  with suffix sorting, dynamic lowest common ancestors and order
  maintenance.

  The technical features of the proposed technique for a given data
  structure~$\DDD$ are the following ones. The new data
  structure~$\DDD'$ is obtained from~$\DDD$ by augmenting the latter
  with an oracle for strings, extending the functionalities of the
  Dietz-Sleator list for order
  maintenance~\cite{DS87,TsakalidisActaInf84}.  The space complexity
  of~$\DDD'$ is $\SSS(n) + O(n)$ memory cells for storing~$n$ keys,
  where $\SSS(n)$ denotes the space complexity of~$\DDD$. Then, each
  operation involving $O(1)$ keys taken from~$\DDD'$ requires
  $O\bigl(\TTT(n)\bigr)$ time, where $\TTT(n)$ denotes the time
  complexity of the corresponding operation originally supported
  in~$\DDD$.  Each operation involving a key~$y$ \emph{not} stored in
  $\DDD'$ takes $O\bigl(\TTT(n) + |y|\bigr)$ time, where $|y|$ denotes
  the length of~$y$.  For the special case where the oracle handles
  suffixes of a string, the achieved insertion time is
  $O\bigl(\TTT(n)\bigr)$.
\end{abstract}

\newpage

\section{Introduction}
\label{sec:introduction}

Many applications manage keys that are arbitrarily long, such as
strings, multi-dimensional points, multiple-precision numbers,
multi-key data, URL~addresses, IP~addresses, XML~path strings and that are modeled either as $k$-dimensional keys for a given
positive integer $k > 1$, or as variable-length keys.  In response to
the increasing variety of these applications, the keys need to be
maintained in sophisticated data structures.  The comparison of any
two keys is more realistically modeled as taking time proportional to
their length, introducing an undesirable slowdown factor in the
complexity of the operations thus supported by the known data
structures.

More efficient \emph{ad hoc} data structures have been designed to
tackle this drawback.  A first version of lexicographic or ternary
search trees~\cite{BS97} dates back to~\cite{Cl64} and is an alternative
to tries.  Each node contains the $i$th symbol of a $k$-dimensional
key along with three branching pointers [left, middle, right] for the
three possible comparison outcomes $[<, =, >]$ against that element.
The dynamic balancing of ternary search trees was investigated with
lexicographic D-trees~\cite{Me79}, multi-dimensional
B-trees~\cite{GK80}, lexicographic globally biased trees~\cite{BST85},
lexicographic splay trees~\cite{ST85}, $k$-dimensional balanced binary
search trees~\cite{Go92}, and balanced binary search trees or $k$BB-trees~\cite{Va96}. Most of these data structures make use of
sophisticated and involved techniques to support search, insert, and
delete of a key of length~$k$ in a given set of $n$ keys, in $O(k +
\log n)$ time~\cite{BST85,Go92}. Some others support also split and
concatenate operations in $O(k + \log n)$
time~\cite{GK80,Me79,ST85,Va96}.  Moreover, other data structures
allow for weighted keys (e.g.~access frequencies) and the $\log n$
term in their time complexity is replaced by the logarithm of the
ratio between the total weights and the weight of the key at
hand~\cite{BST85,Me79,ST85,Va96}.

This multitude of {\em ad hoc\/} data structures stems from the lack
of a general data structural transformation from indivisible
(i.e.~constant-time comparable) keys to strings. Many useful search
data structures, such as AVL-trees, red-black trees~\cite{Ta83},
$(a,b)$-trees~\cite{HM82}, weight-balanced
BB[$\alpha$]-trees~\cite{NR73}, self-adjusting trees~\cite{ST85}, and
random search trees~\cite{AS96}, to name a few, are currently
available. They exhibit interesting combinatorial properties that make
them attractive both from the theoretical and from the practical point
of view.  They are defined on a set of indivisible keys supporting a
total order relation~$<$.  Searching and updating is driven by
constant-time comparisons against the keys stored in them.  Data
structuring designers successfully employ these data structures in
many applications (e.g.~the C++ Standard Template Library~\cite{STL}
or the LEDA package~\cite{LEDA}).  When dealing with keys of
length~$k$, it is natural to see if they can reuse their well-suited
data organizations without incurring in the slowdown factor of $O(k)$
in the time cost for these solutions.

A first step for exploiting the body of knowledge mentioned above,
thus obtaining new data structures for managing strings, has been
presented theoretically in~\cite{GI} and validated experimentally
in~\cite{CGI}.  This general technique exploits the underlying
structure of the data structures by considering the nodes along the
access paths to keys, each node augmented with a pair of integers.  In
order to apply it to one's favorite data structures, the designer must
know the combinatorial properties and the invariants that are used to
search and update those data structures, so as to deal with all
possible access paths to the same node.  This depends on how the
underlying (graph) structure is maintained through the creation and
destruction of nodes and the updates of some internal pointers.  (For
example, we may think of the elementary operations that are performed
by the classical insertion or deletion algorithm for binary search
trees, in terms of the access paths from the root towards internal
nodes or leaves.)  While a general scheme is described for searching
under this requirement, updating is discussed on an individual basis
for the above reason.  A random access path, for example, cannot be
managed unless the possible access paths are limited in number.  Also,
adding an internal link may create many access paths to a given node.
Related techniques, although not as general as that in~\cite{GI}, have
been explored in~\cite{Irving03,Roura01} for specific data structures
being extended to manage strings.

In this paper, we proceed differently. We completely drop any
topological dependence on the underlying data structures and still
obtain the asymptotic bounds of previous results. The goal is to show
that a more general transformation is indeed possible.  We present a
general technique that is capable of reusing many kinds of
(heterogeneous) data structures so that they can operate on strings
and unbounded-length keys.  It is essentially a black-box technique
that just requires that each such data structure, say~$\DDD$, is
driven by constant-time comparisons among the keys (i.e.~no hashing or
bit manipulation of the keys) and that the insertion of a key
into~$\DDD$ identifies the predecessor or the successor of that key
in~$\DDD$. We are then able to transform~$\DDD$ into a new data
structure, $\DDD'$, storing~$n$ strings as keys while preserving all
the nice features of~$\DDD$.

Asymptotically speaking, this transformation is costless. First, the
space complexity of $\DDD'$ is $\SSS(n) + O(n)$, where $\SSS(n)$
denotes the space complexity of~$\DDD$ and the additional $O(n)$ is for the (pointers to the) strings. (We note that the input strings actually occupy memory, but we prefer to view them as external to the
data structure as they are not amenable to changes. Hence, we just store the pointers to
strings, not the strings themselves). Second, each operation involving
$O(1)$ strings taken from~$\DDD'$ requires $O\bigl(\TTT(n)\bigr)$
time, where $\TTT(n)$ denotes the time complexity of the corresponding
operation originally supported in~$\DDD$. Third, each operation
involving a string~$y$ \emph{not} stored in~$\DDD'$ takes
$O\bigl(\TTT(n) + |y|\bigr)$ time, where~$|y|$ denotes the length
of~$y$.

We also consider the special case where the strings are suffixes of a given string. In this
special case if we insert the strings in reverse ordering then all the previously claimed
results hold, and, on top of that, we can implement the insertion operation of a suffix in $O\bigl(\TTT(n)\bigr)$ time.

Our technique exploits the many properties of one-dimensional
searching, and combines techniques from data structures and string
algorithms in a variety of ways.  Formally, we manage input strings
$x_1$, $x_2$,\dots, $x_n$ of total length $M = \sum_{i=1}^n |x_i|$.
Each string $x_i$ is a sequence of $|x_i|$ symbols drawn from a
potentially unbounded~\footnote{Our technique applies to the classical
  comparison-based RAM model adopted in many algorithms for sorting
  and searching. Assuming that $\Sigma$ is unbounded is not a limitation.
  When $\Sigma$ is small, adjacent symbols can be packed into the same
  word of memory. The lexicographic order is preserved by considering
  the words as individual symbols in a larger alphabet, with shorter
  strings consequently.  We do not treat here the case of RAM with
  word size bounded by~$w$ bits.  }  alphabet $\Sigma$, and the last
symbol of $x_i$ is a special endmarker less than any symbol in
$\Sigma$. In order to compare two strings~$x$ and~$y$, it is useful to
employ the length of their longest common prefix, defined as
$\lcp(x,y) = \max \{ \ell \geq 0 \mid x[1 \pp \ell] = y[1 \pp \ell]
\}$ (here, $\ell = 0$ denotes empty prefixes).  Given that length, we
can compare~$x$ and~$y$ in constant time by simply comparing their
first mismatching symbol, which is at position $1+\lcp(x,y)$ in both~$x$
and~$y$.

Keeping this fact in mind, we can use the underlying data
structure~$\DDD$ as a black box. We use simple properties of strings
and introduce a powerful oracle for string comparisons that extends
the functionalities of the Dietz-Sleator
list~\cite{DS87,TsakalidisActaInf84}, which is able to maintain order
information in a dynamic list. We call the
resulting structure a $\dslcp$~list. This data structure stores the sorted input
strings in $O(n)$ memory cells of space and allows us to find the
length of the longest common prefix of any two strings stored in the
$\dslcp$~list, in constant time. We can maintain a dynamic
$\dslcp$~list in constant time per operation (see
Section~\ref{sec:implementation-dslcp} for the operations thus supported) by
using a simple but key idea in a restricted dynamic version of the
range minimum query algorithm~\cite{BenderLCA}. We cannot achieve
constant time per operation in the fully dynamic version of this
problem because we would get the contradiction of comparison-based
sorting in $o(n \log n)$ time by using range minima data structures as
priority queues.

Note that for general strings in the case when $\TTT(n) = \Omega(\log n)$ one can use the following alternative technique. Use a compacted trie with special handling of access to
children (with weight balanced trees), see e.g.~\cite{ColeL03}, in $O(|y| + \log n)$ time in the worst case.  This absorbs the $O\bigl(\TTT(n) + |y|\bigr)$ cost of insertion.
Then on the nodes one can use the dynamic lowest common ancestor queries~\cite{SODA99*235} as the powerful oracle needed. Then it is sufficient to follow Section~\ref{sub:exploiting-lcp-values}, the section in which
one exploits $\lcp$ values in comparison-driven data structures.

Field of interests for our technique include the following scenarios.
\begin{enumerate}
\item
Operations with
sub-logarithmic costs, when $\TTT(n) = o(\log n)$.\footnote{Note that deletions have no such restrictions as insertions.} This happens:
\begin{itemize}
\item
in the worst
case (e.g.~$\DDD$ is a finger search tree~\cite{GHLST87}).
\item
in an amortized sense (e.g.~$\DDD$ is a self-adjusting tree~\cite{ST85}).
\item
with high
probability (e.g.~$\DDD$ is a treap~\cite{AS96}), when considering
access frequencies in the analysis.
\end{itemize}
\item
One desires to use a simpler data structure than the intensive dynamic $\lca$ data structure of~\cite{SODA99*235}.
\item
For the technique for suffixes. Here the running time of the insertion operation is $O\bigl(\TTT(n)\bigr)$ instead of $O\bigl(\TTT(n) + |y|\bigr)$ and, hence, the
above-described method does not hold. Section~\ref{sec:online-construction-suffix-tree} is dedicated to such an application.
\end{enumerate}

We also remark that we do not claim that our technique is as
amenable to implementation in a practical setting. (We suggest to use
the techniques devised in~\cite{GI} and experimented in~\cite{CGI} for
this purpose.)  Nevertheless, we believe that our general technique
may be helpful in the theoretical setting for providing an immediate
benchmark to the data structuring designer.  When inventing a new data
structure for strings, the designer can easily realize whether it
compares favorably to the known data structures, whose functionalities
can be smoothly extended to strings without giving up their structural
and topological properties.

Using our general technique, we obtain previous theoretical bounds in
an even simpler way. We also obtain new results on searching and
sorting strings. For example, we can perform suffix sorting, a crucial
step in text indexing~\cite{MM93} and in block sorting
compression based on the Burrows-Wheeler transform, in
$O\bigl( n + \sum_{i=1}^n \TTT(i) \bigr)$ time, also for unbounded
alphabet~$\Sigma$.
This result is a simple consequence of our result,
when applied to the techniques for one-dimensional keys given, for
example, in~\cite{Me84}. Another example of use is that of storing
implicitly the root-to-nodes paths in a tree as strings, so that we
can support dynamic lowest common ancestor ($\lca$) queries in
constant time, where the update operations involve adding/removing
leaves. In previous work, this result has been obtained with a special
data structure based upon a more sophisticated solution treating also
insertions that split arcs~\cite{SODA99*235}. We obtain a simple
method for a restricted version of the problem.

We present another major contribution of the paper, using our framework for \emph{online indexing}. \emph{Indexing} is one of the most important paradigms in
searching. The idea is to preprocess the text and construct a
mechanism that will later provide answer to queries of the the form
``does a pattern $P$ occur in the text'' in time proportional to the
size of the \emph{pattern} rather than the text. The suffix tree
[7,13,16,17] and suffix array [11,12] have proven to be invaluable
data structures for indexing.

One of the intriguing questions of the algorithms
community is whether there exists a real-time indexing algorithm. An
algorithm is \emph{online} if it accomplishes its task for the $i$th
input without needing the $i + 1$st input. It is \emph{real-time} if,
in addition, the time it operates between inputs is a constant. While
not all suffix trees algorithms are online (e.g. McCreight~\cite{McC76}, Farach~\cite{FOCS::Farach1997})
some certainly are (e.g. Weiner~\cite{Weiner:73}, Ukkonen~\cite{ALGOR::Ukkonen1995}). Nevertheless, the quest for
a real-time indexing algorithm is over 30 years old~\cite{Slisenko78}. It should be
remarked that Weiner basically constructs an online reverse prefix
tree. In other words, to use Weiner's algorithm for online indexing
queries, one would need to reverse the pattern.  For real-time
construction there is some intuition for constructing prefix, rather
than suffix trees, since the addition of a single symbol in a suffix
tree may cause $\Omega(n)$ changes, whereas this is never the case in
a prefix tree.

It should be remarked that for unbounded alphabets, no real-time
algorithm is possible since the suffix tree can be used for
sorting. All known comparison-based online suffix tree construction
algorithms for suffix tree or suffix array construction run in amortized $O(\log n)$ time per symbol and answer search queries in $O(m \log n + tocc)$, e.g. suffix trees, or $O(m + \log n + tocc)$, e.g. suffix arrays. The latter
uses non-trivial pre-processing for LCP (longest common prefix)
queries. However, the best that can be hoped for (but not hitherto
achieved) is an algorithm that pays $\Theta(\log n)$ time for
every \emph{single} input symbol.

The problem of \emph{dynamic indexing}, where changes can be made
anywhere in the text has been addressed as well \cite{SODA::GuFB1994,SICOMP::FerraginaG1998}. Real-time and online indexing
can be viewed as a special case of dynamic indexing, where the changes made are
insertions and deletions at the end (or, symmetrically, the
beginning) of the text. Sahinalp and Vishkin~\cite{SV96} provide a dynamic indexing where
updates are done in time $O(\log^3 n)$. This result was improved by
Alstrup, Brodal and Rauhe~\cite{ABR00} to an $O(\log^2 n \log \log n \log^*
n)$ update time and $O(m + \log n \log \log n + tocc)$ search time. The motivation for
real-time indexing is the case where the data arrives in a constant
stream and indexing queries are asked while the data stream is still
arriving. Clearly a real-time suffix tree construction answers this
need.

Our contribution is the first algorithm for online suffix tree
construction over unbounded alphabets. Our construction has {\bf worst case} $O(\log n)$ time processing
per input symbol, where $n$ is the length of the text
input so far. Furthermore, the search time for a pattern of length $m$
is $O(\min(m \log |\Sigma|, m + \log n))$, where $\Sigma$ is the
alphabet. This matches the best times of the amortized algorithms in the
comparison model. We do so by using a balanced search tree on the suffixes of our text using our proposed technique. This enables insertions
(and deletions) in time $O(\log n)$ and query time of  $O(m+ \log n + tocc)$.
Aside from the balanced search tree itself, the innovative part of the online indexing is the way we
maintain and insert incoming symbols into a suffix tree in time
$O(\log n)$ per symbol. We employ interesting observations that
enable a binary search on the path when a new node needs to be added
to the suffix tree. Note that deletions of characters from the beginning of the text can
also be handled within the same bounds. We note that in the meantime there has been some progress on the problem. Amir and Nor~\cite{AN08} showed a method
achieving real-time pattern matching, i.e. $O(1)$ per character addition, in the case where the alphabet size is of constant size. However, the result (1) does not seem to scale up to a non-constant sized alphabet and (2) does not find the matches of a pattern (rather it announces whether there exists a match of the pattern somewhere in the text when queried). Very recently, new algorithms have been proposed in \cite{BreslauerI11,TsviFOCS2012}.

The paper is organized as follows. In
Section~\ref{sec:gener-techn-strings}, we describe our general
technique, assuming that the $\dslcp$~list is given. We detail
the implementation of the $\dslcp$~list in
Section~\ref{sec:implementation-dslcp}. We discuss an interesting
application in Section~\ref{sec:some-applications}. In Section~\ref{sec:technique-suffixes} we present the technique for suffixes and in Section~\ref{sec:suffix-application} we give a couple of applications.
Finally, we discuss our solutions for online suffix tree construction in Section~\ref{sec:online-construction-suffix-tree}.

\section{The General Technique for Strings}
\label{sec:gener-techn-strings}

We begin with the description of our technique, which relies on an
oracle for strings called the $\dslcp$~list.  As previously mentioned, each string $x$ is a
sequence of symbols drawn from a potentially unbounded alphabet
$\Sigma$, and the last symbol in $x$ is a special endmarker smaller
than any symbol in $\Sigma$.  In order to compare any two strings~$x$
and~$y$, we exploit the length $\ell=\lcp(x,y) = \max \{ \ell \geq 0
\mid x[1 \pp \ell] = y[1 \pp \ell] \}$ of their longest common
prefix. Since we use endmarkers, if $\ell = |x|$, then $\ell = |y|$
and so $x=y$. Otherwise, it is $x < y$ in lexicographic order if and
only if $x[\ell+1] < y[\ell+1]$. Hence, given $\lcp(x,y)$, we can
check $x \leq y$ in constant time.  This is why we center our
discussion around the efficient computation of the $\lcp$~values.

Formally, the $\dslcp$~list stores a \emph{sorted} sequence of strings
$x_1, x_2, \ldots, x_n$ in non-decreasing lexicographic order, where
each string is of unbounded length and is referenced by a pointer
stored in a memory cell (e.g.~\texttt{char} \texttt{*p} in C~language). A
$\dslcp$~list~$L$ supports the following operations:
\begin{itemize}
\item Query $\dslcp(x_p, x_q)$ in~$L$. It returns the value of
  $\lcp(x_p, x_q)$, for any pair of strings $x_p$ and $x_q$ stored
  in~$L$.
\item Insert~$y$ into~$L$. It assigns to~$y$ the position between two
  consecutive keys~$x_{k-1}$ and~$x_k$.  Requirements: $x_{k-1} \leq y
  \leq x_k$ holds, and $\lcp(x_{k-1}, y)$ and $\lcp(y, x_k)$ are given
  along with~$y$ and $x_{k-1}$ (or, alternatively,~$y$ and $x_k$).
\item Remove string~$x_i$ from its position in~$L$.
\end{itemize}

We show in Section~\ref{sec:implementation-dslcp} how to implement
the $\dslcp$~list with the bounds stated below.

\begin{theorem}
  \label{the:dslcp-list}
  A $\dslcp$~list~$L$ can be implemented using $O(n)$ memory cells of
  space, so that querying for~$\lcp$~values, inserting keys into~$L$
  and deleting keys from~$L$ can be supported in $O(1)$ time per
  operation, in the worst case.
\end{theorem}

We devote the rest of this section on how to apply
Theorem~\ref{the:dslcp-list} to a comparison-driven data structure.  The $\dslcp$~list~$L$ is a valid tool
for dynamically computing $\lcp$~values for the strings stored
in~$L$. It is a natural question
to see if we can also exploit~$L$ to compare against an arbitrary
string $y \not \in L$, which we call computation \emph{on the fly} of
the $\lcp$~values since they are not stored in~$L$, nor can be
inferred by accessing $L$. Note that this operation does not
immediately follow from the aforementioned operations, and we
describe in Section~\ref{sub:computing-lcp-values} how to perform them.

After that, we introduce our general technique for data structures in
Section~\ref{sub:exploiting-lcp-values}, where we prove the
main result of this section (Theorem~\ref{the:string-data-structure}), which allows us to reuse a large corpus of
existing data structures as black boxes.

\begin{theorem}
  \label{the:string-data-structure}
  Let~$\DDD$ be a comparison-driven data structure such that the
  insertion of a key into~$\DDD$ identifies the predecessor or the
  successor of that key in~$\DDD$. Then,~$\DDD$ can be transformed
  into a data structure $\DDD'$ for
  strings such that
  \begin{itemize}
  \item the space complexity of $\DDD'$ is $\SSS(n) + O(n)$ for
    storing~$n$ strings as keys (just store the references to strings,
    not the strings themselves), where $\SSS(n)$ denotes the native space
    complexity of~$\DDD$ in the number of memory cells occupied;
  \item each operation involving $O(1)$ strings in $\DDD'$ takes
    $O\bigl(\TTT(n)\bigr)$ time, where $\TTT(n)$ denotes the time
    complexity of the corresponding operation originally supported
    in~$\DDD$;
  \item each operation involving a string~$y$ \emph{not stored} in
    $\DDD'$ takes $O\bigl(\TTT(n) + |y|\bigr)$ time, where $|y|$
    denotes the length of~$y$.
  \end{itemize}
\end{theorem}

\subsection{\boldmath Computing the $\lcp$~values on the fly}
\label{sub:computing-lcp-values}

We now examine the situation in which we have to compare any given
string~$y$ against an arbitrary sequence of strings $x \in L$ (with
some of the latter ones possibly repeated inside the sequence).  If
$y$ has to be compared $g$~times against strings in $L$, it follows a
simple lower bound of $\Omega(g + |y|)$ in the worst case time
complexity, since $y$'s symbols have to be read and we have to produce
$g$~answers. We show how to produce a matching upper bound.  We figure
out ourselves in the worst case situation, namely, the choice of~$x
\in L$ is unpredictable from our perspective. Even in this case, we
show how to compute $\lcp(x, y)$ efficiently. We assume that the empty
string is implicitly kept in~$L$ as the smallest string.

We employ two global variables, $\bestfriend$ and $\bestlcp$, which
are initialized to the empty string and to~0, respectively. During the
computation, the variables satisfy the invariant that, among all the strings
in~$L$ compared so far against~$y$, the one pointed to by $\bestfriend$
gives the maximum $\lcp$~value, and that value is $\bestlcp$. We now
have to compare~$y$ against~$x$, following the simple
algorithm\footnote{Although the code in Figure~\ref{fig:code} can be
  improved by splitting the case $m \geq \bestlcp$ of line~2 into two
  subcases, it does not improve the asymptotic complexity.} shown in
Figure~\ref{fig:code}.  Specifically, using~$L$, we compute~$m =
\dslcp(\bestfriend, x)$, since both strings are in~$L$. If $m <
\bestlcp$, we can infer that $\lcp(x,y) = m$ and return that value.
Otherwise, we may possibly extend the number of matched characters
in~$y$ storing it into $\bestlcp$, thus finding a new $\bestfriend$.
It is a straightforward task to prove the correctness of the invariant
(note that it works also in the border case when $x = \bestfriend$).
It is worth noting that the algorithm reported in
Figure~\ref{fig:code} is a computation of a simplification of a combinatorial property
that has been indirectly rediscovered many times for several string data
structures.

\begin{figure}[t]
\noindent
Compute $\lcp(x, y)$ on the fly:
\begin{algorithmic}[1]
  \STATE $m \PRENDE \dslcp(\bestfriend, x)$
  \IF{$m \geq \bestlcp$}
     \STATE $m \PRENDE \bestlcp$
     \STATE \textsc{while} $x[m+1] = y[m+1]$ \textsc{do} $m \PRENDE m+1$
     \STATE $\bestfriend \PRENDE x$
     \STATE $\bestlcp \PRENDE m$
  \ENDIF
  \STATE \textsc{return} $m$
\end{algorithmic}
\caption{Code for computing $\lcp(x,y)$~values on the fly, where $x
  \in L$ and $y \not \in L$.}
\label{fig:code}
\end{figure}

We now analyze the cost of a sequence of $g$ calls to the code in
Figure~\ref{fig:code}, where $y$ is always the same string while $x$
may change at any different call. Let us assume that the instances
of~$x$ in the calls are $x'_1, x'_2, \ldots, x'_g \in L$, where the
latter strings are not necessarily distinct and/or sorted. For the
given string $y \not \in L$, the total cost of computing $\lcp(x'_1,
y)$, $\lcp(x'_2, y)$,~\dots, $\lcp(x'_g, y)$ on the fly with the code
shown in Figure~\ref{fig:code} can be accounted as follows.  The cost
of invoking the function is constant unless we enter the body of the
while loop at line~4, to match further characters while increasing the
value of~$m$.  We can therefore restrict our analysis to the strings
$x'_i \in L$ that cause the execution of the body of that while loop.
Let us take the $k$th such string, and let $m_k$ be the value of~$m$ at
line~6.  Note that the body of the while loop at line~4 is executed
$m_k - m_{k-1}$ times (precisely, this is true since $\bestlcp =
m_{k-1}$, where $m_0 = 0$).  Thus the cost of computing the
$\lcp$~value for such a string is $O(1 + m_k - m_{k-1})$.

We can sum up all the costs. The strings not entering the while loop
contribute each for a constant number of steps; the others contribute
more, namely, the $k$th of them requires $O(1 + m_k - m_{k-1})$ steps.
As a result, we obtain a total cost of $O(g + \sum_k (m_k - m_{k-1}))
= O(g + |y|)$ time, since $m_k \geq m_{k-1}$ and $\sum_k (m_k -
m_{k-1})$ is upper bounded by the length of the longest matched prefix
of~$y$, which is in turn at most~$|y|$.

\begin{lemma}
  \label{lemma:lcp_on_the_fly}
  The computation on the fly of any sequence of~$g$ $\lcp$~values
  involving a given string $y \not \in L$ and some strings in~$L$ can
  be done in $\Theta(g + |y|)$ time in the worst case.
\end{lemma}

Note that if it was the case that $y \in L$ we could easily obtain a bound
of~$O(g)$ in Lemma~\ref{lemma:lcp_on_the_fly}
using Theorem~\ref{the:dslcp-list}.  However, we assumed that $y \not \in L$. Nevertheless, Lemma~\ref{lemma:lcp_on_the_fly} allows
us to reduce the cost from $O(g \times |y|)$ to $O(g + |y|)$.
Finally, letting $\ell = \max_{1 \leq i \leq g} \lcp(y, x'_i) \leq
|y|$, we can refine the upper and lower bounds of
Lemma~\ref{lemma:lcp_on_the_fly}, obtaining an analysis of $\Theta(g +
\ell)$ time in the worst case.

\subsection{\boldmath Exploiting $\lcp$~values in comparison-driven
  data structures}
\label{sub:exploiting-lcp-values}

We can now finalize the description of our general technique, proving
Theorem~\ref{the:string-data-structure}.  The new data
structure~$\DDD'$ is made up of the original data structure~$\DDD$
along with the $\dslcp$~list~$L$ mentioned in
Theorem~\ref{the:dslcp-list}, and uses the on the fly computation
described in Section~\ref{sub:computing-lcp-values}. The additional
space is that of~$L$, namely, $O(n)$ words of memory.

For the cost of the operations, let us assume first that the generic
operation requiring $\TTT(n)$ time does not change the set of keys
stored in the original~$\DDD$.  When this operation requires to
compare two atomic keys~$x$ and $x'$ in~$\DDD$, we compare the
homologous strings~$x$ and $x'$ in~$\DDD'$ by using~$\dslcp(x, x')$ to
infer whether $x < x'$, $x = x'$, or $x > x'$ holds, in constant time
by Theorem~\ref{the:dslcp-list}. Since there are at most $\TTT(n)$
such comparisons, it takes $O\bigl(\TTT(n)\bigr)$ time.

On the other hand, an operation can also employ a string that
is not stored in~$\DDD'$. When it compares such a string, say~$y$,
with a string~$x$ already in~$\DDD'$, we proceed as in
Section~\ref{sub:computing-lcp-values} to infer the outcome of
comparisons, where $g \leq \TTT(n)$.  By
Lemma~\ref{lemma:lcp_on_the_fly}, the computation takes
$O\bigl(\TTT(n) + |y|\bigr)$ time.~\footnote{Actually, as observed at
  the end of Section~\ref{sub:computing-lcp-values}, it takes
  $O\bigl(\TTT(n) + \bestlcp_F -\bestlcp_I\bigr)$ time, where $\bestlcp_I$
  is the initial value of $\bestlcp$ and $\bestlcp_F$ is its final
  value. This observation is useful when accounting for the complexity
  of a sequence of correlated insertions (as in
  Theorem~\ref{the:suffix-data-structure}), since it yields a
  telescopic sum.}

It remains to discuss the case of operations insertion or deletion of a
string~$y$ into (or out of)~$\DDD'$. Let us consider the insertion of~$y$. While simulating the insertion of $y$ into $\DDD$, at most $\TTT(n)$ keys already in $\DDD'$
have to be compared to $y$, which takes $O(\TTT(n))$ overall time. In the end, we obtain the successor or predecessor of $y$ in $\DDD$, by the hypothesis of Theorem~\ref{the:string-data-structure}.
Using~$L$, we
know both and, therefore, we can compute their $\lcp$~values with~$y$
in $O(|y|)$ time by scanning their first $|y|$ symbols at most, thus, computing the lcp's with their predcessor and successor,
satisfying the requirement for inserting~$y$ into~$L$ in $O(1)$ time.
The final cost is upper bounded by Lemma~\ref{lemma:lcp_on_the_fly},
where $g \leq \TTT(n)$. The deletion of a string~$x$ is much simpler,
involving its removal from~$L$ in $O(1)$ time by
Theorem~\ref{the:dslcp-list}. In summary, the original cost of the
operations in~$\DDD$ preserve their asymptotical complexity in~$\DDD'$
except when new strings $y$ are considered. In that case there is an
\emph{additive} term of $|y|$ in the time complexity. However, this is optimal.


\section{\boldmath Implementation of the $\dslcp$~List}
\label{sec:implementation-dslcp}

We describe how to prove Theorem~\ref{the:dslcp-list}, implementing
the $\dslcp$~list~$L$ introduced in
Section~\ref{sec:gener-techn-strings}.  Recall that  the strings $x_1, x_2, \ldots, x_n$ in~$L$ are in lexicographic
order. We will use the following well known observation.

\begin{lemma}\label{minlcp}{\em \cite{MM93}}
For $p<q$ we have $\lcp(x_p, x_q) = \min \{ \lcp(x_k, x_{k+1}) \mid p \leq k < q \}$.
\end{lemma}

In other words, storing only the $\lcp$~value
between each key~$x_k$ and its successor $x_{k+1}$ in~$L$, for $1 \leq
k < n$, we can answer arbitrary $\lcp$~queries using the so-called
range minimum queries~\cite{BenderLCA}.


We are interested in discussing the dynamic version of the problem. In
its general form, this is equivalent to sorting since it can implement
a priority queue. Fortunately, we can attain constant time per
operation in our case. We consider a special form of insertion and deletion and,
more importantly, we impose the additional constraint that the set of
entries can only vary \emph{monotonically}, which we define shortly.

The type of insertion and deletion that we discuss is as follows. Let $L$ be a list of integers. An {\em insertion} is the replacement of an element of entry $e$ in the list with two new (adjacent) elements $e'$ and $e''$. A {\em deletion} is the replacement of two (adjacent) elements $e'$ and $e''$ with an element $e$.

{\bf Monotonicity:} We say that an insertion or deletion is {\em montone} if  $e = \min\{e',e''\}$.

In our setting the monotonicity constraint is not
artificial, being dictated by the requirements listed in
Section~\ref{sec:gener-techn-strings} when inserting (or deleting) string~$y$
between $x_k$ and $x_{k+1}$. A moment of reflection shows that both
$e' = \lcp(x_k, y)$ and $e'' = \lcp(y, x_{k+1})$ are greater than or
equal to $e = \lcp(x_k, x_{k+1})$, and at
least one of them equals~$e$.

We can therefore reduce our implementation of the $\dslcp$~list to the
following problem.  We are given a dynamic sequence of $m$ values
$e_1, e_2, \ldots, e_m$, and want to maintain the range minima
$\rmq(i,j) = \min \{e_i, e_{i+1}, \ldots, e_j\}$, for $1 \leq i, j
\leq n$, under \emph{monotone} insertions and deletions. A monotone
insertion of $e'$ and $e''$ instead of $e_k$ satisfies $e_k = \min\{e',e''\}$;
a monotone deletion of $e_{k-1}$  and $e_k$ replaced by $e$
satisfies $e = \min\{e_{k-1},e_k\}$ (where we assume $e_0 = e_{m+1} =
-\infty$).  To implement our $\dslcp$~list $L$, we observe that $m =
n-1$ and $e_k = \lcp(x_k, x_{k+1})$, for $1 \leq k \leq n-1$, and that
insertions and deletions in~$L$ correspond to monotone insertions and
deletions of values in our range minima. Note that the insertion of~$y$ can
be viewed as either the insertion of entry~$e'$ to the left of
entry~$e$ (when $e' \geq e'' = e$) or the insertion of~$e''$ to the
right of~$e$ (when $e'' \geq e' = e$).

For some intuition on why monotonicity in a dynamic setting helps consider the following. Suppose we want to build some kind
of tree, whose leaves store the elements $e_1, e_2, \ldots, e_m$ in
left-to-right order. Pick an internal node $u$ that spans, say,
elements $e_i, e_{i+1}, \ldots, e_j$, in its descendant leaves. We
maintain prefix minima $p_k = \min \{e_i, e_{i+1} \ldots, e_k \}$ in $u$
and also suffix minima $q_k = \min \{e_k, e_{k+1}, \ldots,
e_j \}$ for $i \leq k \leq j$. When inserting~$e'$ and $e''$ instead of $e_k$ monotonically, each can change \emph{just two} prefix minima, whereas, in general,
 the prefix minima $p_k, p_{k+1}, \cdots, p_j$ and the suffix minima $q_i, \cdots, q_k$ can all change without our assumption on the
monotonicity. However, with our monotonic assumption, if we insert say $e'$ first between $e_{k-1}$ and $e_{k+1}$, then we insert a new prefix minimum $p = \min \{ p_{k-1}, e' \}$
and $q = \min \{e', q_{k+1} \}$. The same holds for $e''$. We use this fact as a key
observation to obtain constant-time complexity.

For implementing the $\dslcp$~list we adopt a two-level scheme. We
introduce the upper level consisting of the \emph{main tree} in
Section~\ref{sub:main-tree}, and the lower level populated by
\emph{micro trees} in Section~\ref{sub:micro-trees}.  We sketch the
method for combining the two levels in
Section~\ref{sub:combine-levels}.  The net result is a generalization
of the structure of Dietz and Sleator that works for multi-dimensional
keys, without relying on the well-known algorithm of
Willard~\cite{IC::Willard1992} to maintain order in a dense file
(avoiding Willard's data structures for the amortized case has been suggested
in~\cite{BenderEtAl02a}). We focus on insertions since deletions are simpler
and can be treated with partial rebuilding techniques, see e.g.~\cite{Overmars83}, since deletions
replace two consecutive entries with the smallest of the two and so do
not change the range minima of the remaining entries.  When treating
information that can be represented with $O(\log n)$ bits, we will
make use of table lookups in $O(1)$ time. The table lookup idea is explained below. The reader may verify that
we also use basic ideas from previous
work~\cite{ArgVit03,BenderLCA,Harel1}.

\subsection{Main tree}
\label{sub:main-tree}

For the basic shape of the main tree we follow the approach of
weight-balanced trees~\cite{ArgVit03}. The main tree has $m$ leaves (for now we assume $m=n+1$, all the $\lcp$s including the two at either end, later we will
have a smaller main tree and $m$ will be smaller, see Section~\ref{sub:combine-levels}), all on the same level (identified as level $0$), each leaf containing
one entry of our range minima problem. The \emph{weight} $w(v)$ of a node
$v$ is \textit{(i)} the number of its children (leaves), if $v$ is on
level $1$ \textit{(ii)} the sum of the weights of its children, if $v$
is on a level $l>1$.  Let $b>4$ be the \emph{branching parameter}, so
that $b = O(1)$.  We maintain the following constraints on the weight
of a node $v$ on a level $l$.
\begin{enumerate}
\item \label{item:1} If $l=1$, $b\le w(v)\le 2b-1$.
\item \label{item:2} If $l>1$, $w(v)< 2b^l$.
\item \label{item:3} If $l>1$ and $v$ is not the root of the tree,
  $w(v)>\frac{1}{2}b^l$.
\end{enumerate}

From the above constraints it follows that each node on a level $l>1$
in the main tree has between $b/4$ and $4b$ children (with the
exception of the root that can have a minimum of two children).  From
this we can easily conclude that the height of the main tree is $h =
O(\log_b m)=O(\log m)$.

When a new entry is inserted as a new leaf $v$ in the main tree, any
ancestor~$u$ of~$v$ that does not respect the weight constraint is
split into two new nodes and the new child is inserted in the parent
of~$u$ (unless~$u$ is the root, in which case a new root is created).
That rebalancing method has an important property.
\begin{lemma}[\cite{ArgVit03}]
  \label{lem:reb}
  After splitting a node $u$ on level $l>1$ into nodes $u'$ and
  $u''$, at least $b^l/2$ inserts have to be performed below $u'$ (or
  $u''$) before splitting again.
\end{lemma}

The nodes of the main tree are augmented with two secondary
structures.

The first secondary structure is devoted to $\lca$~queries.  Each
internal node $u$ is associated with a numeric identifier $\sibnum(u)$
representing its position among its siblings; since the maximum number
of children of a node is a constant, we need only a constant number of
bits, say~$c_b$, to store each identifier.  Each leaf $v$ has two
vectors associated, $\pathstring_v$ and $\ancestors_v$. Let $u_i$ be
the $i$th ancestor of $v$ (starting from the root): the $i$th location
of $\pathstring_v$ contains the identifier $\sibnum(u_i)$, and the
$i$th location of $\ancestors_v$ contains a pointer to $u_i$.  Note
that $\pathstring_v$ occupies $h \times c_b = O(\log m)$ bits and so
we can use table lookup, which we will shortly explain.  These auxiliary vectors are used to find the
$\lca$ between any two leaves $v'$, $v''$ of the main tree in constant
time.  First, we find $j=\lcp(\pathstring_{v'}, \pathstring_{v''})$ by
table lookups; then, we use the pointer in $\ancestors_v[j]$ to access
the node.

The table lookup method, in our case to find the $\lcp$ of strings that are of $O(\log m)$ length is done as follows.
First recall that $m \leq n$ and hence the string length of $\pathstring_v$ is also $O(\log n)$ and fits into a constant number of words.
Now, let us represent $\pathstring_v$
in a few consecutive words, $\pathstring_v^1$, $\pathstring_v^2, \ldots, \pathstring_v^c$  so that we use exactly $(\log n)/3$ bits
of each word (we could have chosen a different small constant to divide $\log n$). This still keeps the representation
in a constant number of words. Now, to find the $\lcp$ between $\pathstring_{v'}$ and $\pathstring_{v''}$ we compare
$\pathstring_{v'}^1$ with $\pathstring_{v''}^1$ etc. till we find $\pathstring_{v'}^i \not= \pathstring_{v''}^i$. Now, if we can find the $\lcp$ between
$\pathstring_{v'}^i$ and $\pathstring_{v''}^i$ the $\lcp$ between $\pathstring_{v'}$ and $\pathstring_{v''}$ will be computable. However,
there are only $(\log n)/3$ bits used in $\pathstring_{v'}^i$ and $\pathstring_{v''}^i$. So, we can precompute an $\lcp$ table for
all the $(2^{(\log n) / 3} = n^{1/3})$ x $(2^{\log n / 3} = n^{1/3})$ possible pairs of strings that we may have. This table will be of size $n^{2/3}$
and, hence, can be precomputed in a straightforward manner. Now all we need to do is to directly access the correct table value by addressing it with
$\pathstring_{v'}^i$ and $\pathstring_{v''}^i$ and we are done.\footnote{Alternatively, we can perform the exclusive bitwise or and find the most significant bit set to~1.}

The second secondary structure is devoted to maintain some range
minima.  Each internal node $u$ has an associated a doubly linked
list~$\lcplist_u$ that contains a copy of all the entries in the
descendant leaves of~$u$. The order in~$\lcplist_u$ is
identical to that in the leaves (i.e.~the lexicographical order in
which the strings are maintained).  As previously mentioned, we
maintain the prefix minima and the suffix minima in~$\lcplist_u$.
We also keep the minimum entry of $\lcplist_u$. Its
purpose is to perform the following query in $O(b)=O(1)$ time: given any
two siblings $u'$ and $u''$, compute the minimum of the entries stored
in the $\lcplist$s of the siblings between $u'$ and $u''$ (excluded).
Finally, we associate with each leaf~$v$ a vector~$\clones_v$ containing
pointers to all the copies of the entry in~$v$, each copy stored in
the doubly linked lists $\lcplist$s of $v$'s~ancestors.

Because of the redundancy of information, $O(m\log m)$ words of memory
is the total space occupied by the main tree, but now we are able to
answer a general range minimum query $\rmq(i,j)$ for an interval $[i
\pp j]$ in constant time.  We first find the lowest common
ancestor~$u$ of the leaves~$v_i$ and~$v_j$ corresponding to the~$i$th
and the~$j$th entries, respectively.  Let~$u_i$ be the child of~$u$
leading to~$v_i$ and~$u_j$ the child of~$u$ leading to~$v_j$ (they
must exist and we can use~$\pathstring_{v_i}$ and~$\pathstring_{v_j}$
for this task). We access the copies of the entries of~$i$ and~$j$
in~$\lcplist_{u_i}$ and~$\lcplist_{u_j}$, respectively,
using~$\clones_{v_i}$ and~$\clones_{v_j}$. We then take the suffix
minimum anchored in~$i$ for~$\lcplist_{u_i}$, and the prefix minimum
anchored in~$j$ for~$\lcplist_{u_j}$.  We also take the minima in the
siblings between~$u_i$ and~$u_j$ (excluded). The minimum among these
$O(1)$ minima is then the answer to our query for interval $[i \pp
j]$.

\begin{lemma}
  \label{lemma:main-tree}
  The main tree for $m$ entries occupies $O(m \log m)$ space, and
  support range minima queries in $O(1)$ time and monotone updates in
  $O(\log m)$ time.
\end{lemma}

In order to complete the proof of Lemma~\ref{lemma:main-tree}, it
remains to see how the tree can be updated. We are going to give a
``big picture'' of the techniques used, leaving the standard details
to the reader.  We already said that deletions can be treated lazily
with the standard partial rebuilding technique. We follow the same
approach for treating the growth of the height~$h$ of the main tree
and the subsequent variations of its two secondary
structures,~$\pathstring$,~$\ancestors$,~$\lcplist$, and~$\clones$.
From now on, let us assume w.l.o.g.\mbox{} that the insertions do not
increase the height of the main tree.

When a new element $e$ is inserted, we know by hypothesis a pointer to
its predecessor $e_k$ (or successor $e_{k+1}$) and the pointer to the
leaf $v$ of the main tree that receives a new sibling~$v'$ and
contains the ($\lcp$) value to be changed.  The creation and
initialization of the vectors associated with the new leaf~$v'$ can be
obviously done in $O(\log m)$ time.  Then we must propagate the
insertion of the new entry in~$v'$ to its ancestors.  Let~$u$ be one
of these ancestors. We insert the entry into its position
in~$\lcplist_u$, using~$\clones_{v}$, which is correctly set (and
useful for setting~$\clones_{v'}$).  As emphasized at the beginning of
Section~\ref{sec:implementation-dslcp}, the monotonicity guarantees
that the only prefix minima changing are constant in number and near
to the new entry (an analogous situation holds for the suffix minima).
As long as we do not need to split an ancestor, we can therefore
perform this update in constant time per ancestor.

If an ancestor~$u$ at level~$l$ needs to split in two new nodes $u'$
and~$u''$ so as to maintain the invariants on the weights, there may be many
$\pathstring_v$, $\ancestors_v$ that need to be changed. However, a careful check verifies that there are $O(b^l)$ such values that need to be changed. Moreover,
we need to
recalculate $O(|\lcplist_u|)=O(b^l)$ values of prefix and suffix
minima. By Lemma~\ref{lem:reb} we can immediately conclude that the
time needed to split~$u$ is $O(1)$ in an amortized sense. A lazy approach
to the construction of the lists~$\lcplist_{u'}$ and~$\lcplist_{u''}$ and then to the $O(b^l)$ $\pathstring_v$'s and  $\ancestors_v$'s that need to be changed
will lead to the desired worst case constant time complexity for the
splitting of an internal node. This construction is fairly technical
but follows closely what has been introduced in~\cite{ArgVit03}.

\subsection{Micro trees for indirection}
\label{sub:micro-trees}

We now want to reduce the update time of Lemma~\ref{lemma:main-tree}
to $O(1)$ in the worst case and the space to $O(m)=O(n)$ words of
memory.  We use indirection~\cite{BenderEtAl02a} by storing
$O(m/\log^2n)$ lists of entries, each of size $O(\log^2 n)$, so that
their concatenation provides the list of entries stored in the leaves
of the main tree described in Section~\ref{sub:main-tree}. We call
buckets these small lists and use the following result as
in~\cite{DS87}.

\begin{lemma}[\cite{ACTAI::LevcopoulosO1988}]
  \label{lemma:levover}
  If the largest bucket is split after every other $k$ insertions into
  any buckets, then the size of any bucket is always $O(k\log n)$.
\end{lemma}

Using Lemma~\ref{lemma:levover} with $k=\log n$, we can guarantee that
the small lists will always be of size $O(\log^2 n)$. We now focus on
how to store one of them in $O(\log^2 n)$ space, such that the overall
space is $O(m)=O(n)$ words (plus $o(n)$ for some shared lookup tables)
and the following constant-time operations are supported: insertion,
deletion, split, merge, range minima (thus solving also prefix minima and
suffix minima).

Each small list is implemented using one micro tree and $O(\log n)$
succinct Cartesian trees plus $O(1)$ lookup tables of size $o(n)$ that
are shared among all the small lists, thus preserving asymptotically
the $O(n)$ overall space bound.

A micro tree satisfies the invariants~\ref{item:1}--\ref{item:3} of
the main tree in Section~\ref{sub:main-tree}, except that each node contains $O(\log n/\log \log n)$
entries and the fan out is now $\Theta(\log n/ \log\log n)$ instead of
$O(1)$. Within each node, we also store a sorted array of $O(\log
n/\log \log n)$ pointers to the entries in that node: since each
pointer requires $O(\log \log n)$ bits, we can fit this array in a
single word and keep it sorted under operation insertion, deletion,
merge and split, in $O(1)$ time each (still table lookups).

We guarantee that a micro tree stores the $O(\log^2 n)$ entries in
$O(1)$ levels, since its height is $h = O(1)$. To see why we partition
these entries into sublists of size $O(\log n/\log\log n)$, which form
the leaves of the micro tree.  Then, we take the first and the last
entry in each sublist, and copy these two entries into a new sublist
of size $O(\log n/\log \log n)$, and so on, until we form the root.
Note that a micro tree occupies $O(\log^2 n)$ words of memory.

Each node of the micro tree has associated a succint Cartesian
tree~\cite{BenderLCA,Harel1,Fischer:2011:SEP}, which locally supports
range minima, split, merge, insert and delete in $O(1)$ time for the
$O(\log n)$ keys. The root of the Cartesian tree is the minimum entry
and its left (right) subtree recursively represents the entries to the
left (right). The base case corresponds to the empty set which is
represented by the null pointer. The observation in~\cite{BenderLCA}
is that the range minimum from entry~$i$ to entry~$j$ is given by the
entry represented by $\lca(i, j)$ in the Cartesian tree.  The most
relevant feature is that the range minimum can be computed without
probing any of the $O(\log n/\log \log n)$ entries of the micro-tree
node $u$. Indeed, by solely looking at the tree topology of that
Cartesian tree for $u$ \cite{Fischer:2011:SEP}, the preorder number of
the $\lca$ gives the position of the entry in $u$. Since the Cartesian
tree has size $O(\log n/\log \log n)$, we can succinctly store it in a
single word using only $O(\log n/\log \log n)$ bits. Using a set of
$O(1)$ lookup tables and the Four Russian trick, it is now a standard
task to locate a range minima, split into or merge two succinct
Cartesian trees, insert or delete a node in it, in $O(1)$ time each.

Note that the succinct Cartesian tree gives us a position $r$ of an
entry in the micro-tree node $u$. We therefore need to access, in
constant time, the $r$th entry in $u$ after that, and the sorted array
of $O(\log n/\log \log n)$ pointers inside $u$ is aimed at this goal.

\subsection{Implementing the operations (Theorem~\ref{the:dslcp-list})}
\label{sub:combine-levels}

In order to prove Theorem~\ref{the:dslcp-list}, we adopt a high-level
scheme similar to that of Dietz-Sleator lists~\cite{DS87}.  The main
tree has $m = O(n / \log^2 n)$ leaves. Each leaf is associated with a
distinct bucket (see Section~\ref{sub:micro-trees}), so that the
concatenation of these buckets gives the order kept in the
$\dslcp$~list. Each bucket contributes to the main tree with its
leftmost and rightmost entries (actually, the range minima of its two
extremal entries).  To split the largest bucket every other $\log n$
insertions, we keep pointers to the buckets in a queue sorted by
bucket size. We take the largest such bucket, split it in $O(1)$ time
and insert two new entries in the main tree. Fortunately, we can
perform incrementally and lazily the $O(\log n)$ steps for the
insertion (Lemma~\ref{lemma:main-tree}) of these two entries before
another bucket split occur. At any time only one update is pending in
the main tree by an argument similar to that in~\cite{DS87}.

\section{An Application of the Technique}
\label{sec:some-applications}

We now describe an application of Theorems~\ref{the:dslcp-list}
and~\ref{the:string-data-structure} in other areas.

\subsection{\boldmath Dynamic lowest common ancestor ($\lca$)}
\label{sub:dynamic-LCA}

The lowest common ancestor problem for a tree is at the heart of
several algorithms~\cite{BenderLCA,Harel1}. We consider here the
dynamic version in which insertions add new leaves as children to
existing nodes and deletions remove leaves. The more general (and
complicated) case of splitting an arc by inserting a node in the
middle of the arc is treated in~\cite{SODA99*235}.

We maintain the tree as an Euler tour, which induces an implicit
lexicographic order on the nodes. Namely, if a node is the $i$th child
of its parent, the implicit label of the node is~$i$. The root has
label~0. (These labels are mentioned only for the purpose of
presentation.) The implicit string associated with a node is the
sequence of implicit labels obtained in the path from the root to that
node plus an endmarker that is different for each string (also when
the string is duplicated; see the discussion below on insertion).
Given any two nodes, the $\lcp$~value of their implicit strings gives
the string implicitly represented by their~$\lca$. We maintain the
Euler tour with a $\dslcp$~list~$L$ in $O(n)$ space (see
Section~\ref{sec:gener-techn-strings}), where $n$ is the number of nodes (the
strings are implicit and thus do not need to be stored).  We also
maintain the dynamic data structure in~\cite{AH} to find the level
ancestor of a node in constant time.

Given any two nodes~$u$ and~$v$, we compute~$\lca(u,v)$ in constant
time as follows. We first find $d = \lcp(s_u, s_v)$ using~$L$, where
$s_u$ and $s_v$ are the implicit strings associated with~$u$ and~$v$,
respectively. We then identify their ancestor at depth~$d$ using a
level ancestor query.

Inserting a new leaf duplicates the implicit string~$s$ of the leaf's
parent, and puts the implicit string of the leaf between the two
copies of~$s$ thus produced in the Euler tour. Note that we satisfy
the requirements described in Section~\ref{sec:gener-techn-strings} for the
insert, as we know their $\lcp$ values. By
Theorem~\ref{the:dslcp-list}, this takes $O(1)$ time. For a richer
repertoire of supported operations in constant time, we refer
to~\cite{SODA99*235}.

\begin{theorem}
  \label{the:dynamic-lca}
  The dynamic lowest common ancestor problem for a tree, in which
  leaves are inserted or removed, can be solved in $O(1)$ time per
  operation in the worst case, using a $\dslcp$~list and the
  constant-time dynamic level ancestor.
\end{theorem}

\section{The General Technique for Suffixes}
\label{sec:technique-suffixes}

We now consider a special case of Theorem~\ref{the:string-data-structure} when the strings are limited to be the suffixes of the same string.

\begin{theorem}
  \label{the:general-suffix-data-structure}
  Let~$\DDD$ be a comparison-driven data structure such that the
  insertion of a key into~$\DDD$ identifies the predecessor or the
  successor of that key in~$\DDD$. Then,~$\DDD$ can be transformed
  into a data structure $\DDD'$ for
  suffixes of a string $s$ of length $n$ such that
  \begin{itemize}
  \item the space complexity of $\DDD'$ is $\SSS(n) + O(n)$ for
    storing~$n$ suffixes as keys (just store the references to the suffixes,
    not the suffixes themselves), where $\SSS(n)$ denotes the native space
    complexity of~$\DDD$ in the number of memory cells occupied;
  \item each operation involving $O(1)$ suffixes in $\DDD'$ takes
    $O\bigl(\TTT(n)\bigr)$ time, where $\TTT(n)$ denotes the time
    complexity of the corresponding operation originally supported
    in~$\DDD$;
  \item an insertion operation of suffix $ay$, \emph{not stored} in
    $\DDD'$, where $a$ is a character and all the suffixes of $y$ are stored in $\DDD'$
   takes $O\bigl(\TTT(n)\bigr)$ time.
  \item each operation (other than the insertion mentioned) involving a string~$y$ \emph{not stored} in
    $\DDD'$ takes $O\bigl(\TTT(n) + |y|\bigr)$ time, where $|y|$
    denotes the length of~$y$.
  \end{itemize}
\end{theorem}

Theorem~\ref{the:general-suffix-data-structure} is very similar to Theorem~\ref{the:string-data-structure}. Hence, the correctness is as well. The difference between the two is the claimed running time for the insertion operation.
We now turn to showing that the claimed running time of the insertion operation of Theorem~\ref{the:general-suffix-data-structure} can indeed be implemented.

Note the special requirement of the insertion operation; when inserting $ay$ into $L$ all suffixes of $y$ must already be in $L$. In other words, a careful insertion of suffixes in reverse order is required. This will be used
to achieve the desired time bound.

To do so, we first augment the $\dslcp$~list~$L$ of
Theorem~\ref{the:string-data-structure} with suffix links. A suffix link
$\slink(ay)$ points to the position of $y$ inside~$L$.
For sake of completeness we assume that the empty string, $\epsilon$, is always in $L$ and that $\slink(\sigma)$ points to $\epsilon$, where $\sigma$ is the suffix of length one.


Before inserting~$ay$ into $\DDD'$, the suffix links are defined for every suffix of $y$ (including $y$ itself).
The current entry in~$L$ is~$y$. So, it is immediate to set up $\slink(ay)$. The predecessor $p_{ay}$
and the successor $s_{ay}$ of $ay$ will be found using the algorithm of $\DDD$. However, we desire to achieve time $O\bigl(\TTT(n)\bigr)$ for this insertion.
Hence, the challenge is to implement the algorithm in $O(1)$ time per comparison. Also, the $\lcp$s $\lcp(p_{ay},ay)$ and $\lcp(ay,s_{ay})$ need to be computed.

Each time $ay$ is compared to a suffix $x$ we can directly evaluate whether $x$ begins with an $a$ (in constant time). If it begins with $\sigma$ different from $a$ then we immediately
know (from the comparison) the lexicographical ordering between $x$ and $ay$ (and that $\lcp(x,ay)=0$). If $x$ begins with $a$, i.e. $x=az$, then $\slink(x)$ points to $z$. Since both $y$ and $z$ are suffixes in $\DDD'$ then by
Theorem~\ref{the:dslcp-list} in $O(1)$ time we can compute $\lcp(y,z)$, which implies that we can compute $\lcp(ay,x)=\lcp(y,z)+1$ in $O(1)$ time. The characters at location $\lcp(ay,x)+1$ of the suffixes
$ay$ and $x$ are sufficient to determine the lexicographic ordering of the two strings. Hence, the lexicographic ordering and the $\lcp$ of $ay$ with any other suffix in $L$ can be computed in $O(1)$ time.

\section{Suffix Technique Applications}
\label{sec:suffix-application}

We now describe a couple of applications of Theorems~\ref{the:dslcp-list}
and~\ref{the:suffix-data-structure}.

\subsection{Suffix sorting}
\label{sub:suffix-sorting}

Suffix sorting is very useful in data compression, e.g. Burrows-Wheeler transform~\cite{BW94}, and in text indexing
(suffix arrays~\cite{MM93}). The computational problem is, given an
input string~$T$ of length~$n$, how to sort lexicographically the
suffixes of~$T$ efficiently.  Let $s_1$, $s_2$,~\ldots, $s_n$ denote
the suffixes of~$T$, where $s_i = T[i \pp n]$ corresponds to the
$i$th suffix in~$T$.

\begin{theorem}
  \label{the:suffix-data-structure}
  Let~$\DDD'$ be a data structure for managing suffixes obtained
  following Theorem~\ref{the:suffix-data-structure}. Then, all the
  suffixes of an input string of length~$n$ can be inserted
  into~$\DDD'$, in space $O(n)+\SSS(n)$ and time
  \[
  O\left(n+\sum_{i=1}^n \TTT(i) \right),
  \]
  where $\TTT(\cdot)$ denotes the time complexity of the insert operation
  in the original data structure~$\DDD$ from which~$\DDD'$ has been
  obtained. The suffixes can be retrieved in lexicographic order in
  linear time.
\end{theorem}

\begin{proof}
The proof follows directly from Theorem~\ref{the:general-suffix-data-structure}.
\end{proof}

\subsection{Balanced indexing structure}
\label{section:binary-indexing-structure}


Following Theorem~\ref{the:general-suffix-data-structure} we define the {\em balanced indexing structure}, shorthanded to BIS. The BIS handles strings and
the underlying structure is a balanced search tree, i.e. $\DDD$ is a standard balanced search tree (of your choice) and $\DDD'$ now is a binary search tree, where the elements
are suffixes of an input string. However, since we may continue inserting suffixes dynamically, the scenario is of a balanced search tree over an online
text $T$. We do point out that the data structure assumes that the text is received from right to left (however, this has no bearings on the online setting as we can always virtually flip the text and query patterns).
In fact, the correct way of viewing this scenario is that of an online indexing scenario, see next subsection.
The BIS has
proven to be instrumental in some other indexing data structures such
as the Suffix Trists~\cite{CKL06} and various heaps of strings~\cite{KLS07}.

\subsubsection{BIS as an Online Indexing Data Structure}

Given a pattern $P=p_1 p_2\cdots p_m$ one desires to find all occurrences of $P$ in $T$ using the BIS of $T$. Using Theorem~\ref{the:general-suffix-data-structure} one can use the balanced search tree to find
$P$ in $O(m+\TTT(n))$ time. In the case of (most) balanced search trees $\TTT(n)=\log n$. Moreover, using the balanced search tree one can find the predecessor and successor of $P$ in the same time. The two suffixes returned define the
interval in the $\dslcp$~list of all of the appearances of $P$. This sublist can be scanned in $O(tocc)$ time using the $\dslcp$~list. Hence,

\begin{theorem}
  \label{the:BIS}
  A BIS can be implemented in $O(n)$ memory cells, so that
  \textit{Addition} and \textit{Delete} operations take $O(\log n)$
  time each, and \textit{Query($P$)} takes $O(|P| + \log n)$ time plus $O(1
  + \mathit{tocc})$ time for reporting all the $\mathit{tocc}$
  occurrences of~$P$.
\end{theorem}

This gives an online indexing scheme with times equivalent to the static suffix array~\cite{MM93}.

\section{Online Construction of Suffix Trees}
\label{sec:online-construction-suffix-tree}

In this section we will show a major application of the results described above. Specifically, we will be using the BIS to achieve the results of this section.

Our goal will be to achieve online suffix trees with quick worst case update time. We do so in
$O(\log n)$ worst case time per insertion or deletion of a character
to or from the beginning of the text.

\subsection{Suffix tree data}
\label{sub:suffix-tree-data}

We will first describe the relevant information maintained within each
inner node in the suffix tree. Later we will show how to maintain this data over
insertions and deletions. For a static text, each node in the
suffix tree has a maximum outdegree of $|\Sigma|$, where $\Sigma$ is
the size of the alphabet of the text string $T$ (each outgoing edge
represents a character from $\Sigma$, and each two outgoing
edges represent different characters). An array of
size $|\Sigma|$ for each node can be maintained to represent the outgoing edges (a
non-existing edge can be represented by NIL), and then, a traversal of the suffix tree with a pattern spends constant time
at each node searching for the correct outgoing edge. However, the suffix tree would have size $O(n|\Sigma|)$ which is not
linear. Moreover, for an online construction, we cannot guarantee that the alphabet of the text will remain the same (in
fact, the alphabet of the text can change significantly with time). Therefore, we use a balanced search tree for each node in the
suffix tree (not to be confused with the BIS).  Each such balanced
search tree contains only nodes corresponding to characters of
outgoing edges. The space is now linear, but it costs
$O(\log|\Sigma|)$ time to locate the outgoing edge of a node. Nevertheless, in the
on-line scenario, we use this solution as balanced search
trees allow us to insert and delete edges dynamically (in $O(\log
|\Sigma |)$ time). Therefore, the cost of adding or
deleting an outgoing edge is $O(\log|\Sigma|)$ where $|\Sigma|$ is the
size of the alphabet of the string at hand. The time for locating an
outgoing edge is also $O(\log|\Sigma|)$. Note that we always have
$|\Sigma|\leq n$. Hence, if during the process of an addition or a
deletion we insert or remove a constant number of outgoing edges in
the suffix tree (as is the case), the cost of insertion and removal is
$O(\log|\Sigma|)$ which is within our $O(\log n)$
bound. In addition, for each node $u$ in the suffix tree we maintain the
length of the string corresponding to the path from the root to $u$. We denote this length by
$\mathit{length}(u)$.

We note that many of the operations on suffix trees (assuming linear space is desired) use various pointers to the text in order to save
space for labeling the edges. We will later show how to maintain such
pointers, called \emph{text links}, within our time and space
constraints.  We also note that a copy of the text saved in
array format may be necessary for various operations, requiring direct
addressing. As mentioned before, this can be done with constant time
update by standard de-amortization techniques.

\subsection{Finding the entry point}
\label{sub:finding-the-entry-point}

We now proceed to the online construction of the suffix tree. Assume
we have already constructed the suffix tree for string $T$ of size
$n$, and we are interested in updating the suffix tree so it will
be the suffix tree of string $aT$ where $a$ is some character.
It is a known fact that a depth first search (DFS) on the suffix tree
encounters the leaves, which correspond to suffixes, in lexicographic
order of the suffixes. Hence, the leaves of the suffix tree in the
order encountered by the DFS form the lexicographic ordering of
suffixes, which is in fact maintained by the BIS in the $\dslcp$~list
(see Section~\ref{section:binary-indexing-structure}). So, upon
inserting suffix $aT$ into the tree, $\mathit{node}(aT)$ needs to be inserted as a leaf. We know between which two leaves of the suffix tree
$\mathit{node}(aT)$ will be inserted according to the lexicographic ordering of the suffixes.

The insertion of the new suffix is implemented by either adding a new leaf as a child of an existent inner
node in the suffix tree, or by splitting an edge, adding a new node $u$ on the edge, and then the new leaf is a child of $u$.
We define the \emph{entry point} of $\mathit{node}(aT)$ as follows. If $\mathit{node}(aT)$ is inserted
as a child of an already existing node $u$, then $u$ is the entry point of $\mathit{node}(aT)$. If $\mathit{node}(aT)$
is inserted as a child of a new node that is inserted while splitting an edge $e$, then $e$ is the entry point of $\mathit{node}(aT)$.

We assume without loss of generality that $\mathit{node}(aT)$ is inserted between $\mathit{node}(T')$ and
$\mathit{node}(T'')$, where $aT$ is lexicographically bigger than $T'$ (hence
$\mathit{node}(aT)$ appears directly after $\mathit{node}(T')$ in the $\dslcp$~list), and $aT$ is
lexicographically smaller than $T''$ (hence $\mathit{node}(T'')$ appears directly before
$\mathit{node}(aT)$ in the $\dslcp$~list). We denote by $x$ the lowest common ancestor of $\mathit{node}(T')$ and $\mathit{node}(T'')$ in the suffix tree of $T$.
Consider the two paths $P_{T'}$ and $P_{T''}$ from $x$ to $\mathit{node}(T')$ and $\mathit{node}(T'')$ respectively. Clearly these two
paths, aside from $x$, are disjoint. We now prove the following lemma that will later assist us in finding the entry point for $\mathit{node}(aT)$.


\begin{lemma}\label{lem:entry_point}
The entry point of $\mathit{node}(aT)$ is on $P_{T'} \cup P_{T''}$. Furthermore, we can decide in constant time which of the following is correct:

\begin{enumerate}
\item{} $x$ is the entry point.
\item{} The entry point is in $P_{T'}-\{x\}$.
\item{} The entry point is in $P_{T''}-\{x\}$.

\end{enumerate}
\end{lemma}

\begin{proof} It follows from $\lcp$ properties that $\lcp(aT,T')\geq \lcp(T',T'')$, and that $\lcp(aT,T'')\geq \lcp(T',T'')$. Thus, from the connection between the lowest common ancestor of two leaves in the suffix tree, and the longest common prefix of the suffixes corresponding to those two leaves, we conclude that the entry point must be in the subtree of $x$ in the suffix tree, as $label(x)$ is a prefix of $aT$. If the entry point is not in $P_{T'} \cup P_{T''}$ then the entry node is on some path $P$ branching out of $P_{T'} \cup P_{T''}$ ending in some leaf $\ell$. However, this would imply that the suffix corresponding to $\ell$ is lexicographically between $T'$ and $T''$, contradicting the fact that they are neighboring suffixes in the suffix tree prior to the insertion.

Hence:
\begin{enumerate}
\item{} If $\lcp(aT,T') = \lcp(aT,T'')$ then $x$ is the entry point.
\item{} If $\lcp(aT,T') > \lcp(aT,T'')$ then the entry point is in $P_{T'}-\{x\}$.
\item{} If $\lcp(aT,T') < \lcp(aT,T'')$ then the entry point is in $P_{T''}-\{x\}$.

\end{enumerate}

The comparison can be performed in constant time using the $\dslcp$~list. Note that even if $aT$ is not yet in the $\dslcp$~list we can compute $\lcp(aT,T')$, by comparing the first character of $T'$ with $a$. If they are different then the $\lcp$ value is 0. Otherwise, we compare the $\lcp$ of $T$, which is accessible as the last suffix inserted, and $\slink(T')$ (both are in the list). Then $\lcp(aT,T') = \lcp(T,\slink(T'))+1$.
\end{proof}

If $x$ is the entry point, then we can easily insert $\mathit{node}(aT)$ as a new child, with the edge labeled starting with the symbol
corresponding to the $k$'th character of $aT$ where $k=\lcp(T',T'')+1$. The insertion requires $O(\log |\Sigma|)$ time.

The other two cases  ($\lcp(aT,T') > \lcp(aT,T'')$ and $\lcp(aT,T') < \lcp(aT,T'')$) are symmetric. Hence, without loss of generality we
assume that $\lcp(aT,T') > \lcp(aT,T'')$. In this case, $\mathit{node}(aT)$ and $\mathit{node}(T')$ share a common path from the root of
the suffix tree until the entry point, and this path corresponds to $\lcp(aT,T')$, as it is the length of the labels on the joint path.
So, our goal is to find a node $v\in P_{T'}$ such that $|label(v)|\leq \lcp(aT,T')$ of maximal depth in the suffix tree. This is discussed next.

Note that once we find $v$ there is not much work left to be done. Specifically, the only nodes or edges which might change are $v$ and
its outgoing edges. This is because $\mathit{node}(aT)$ will enter either as a child of $v$ (in which case $v$ is the entry point of the
new suffix), or one of $v$'s outgoing edges will have to break into two as described in the previous subsection (in which case that edge is the
entry point). We can easily distinguish between the two options by noting that if $k=\lcp(T',aT)$ (we can
calculate this in constant time through the $\dslcp$~list, as in the end of the proof of Lemma~\ref{lem:entry_point}) then $v$ will be the parent of $\mathit{node}(aT)$, and if
there is an inequality ,then we must break an edge for the parent of
$\mathit{node}(aT)$. Each of these cases will take at most
$O(\log |\Sigma|)$ time. So we are left with the task of finding $v$.

The following is derived directly from Lemma~\ref{minlcp}.

\begin{corollary}\label{cor:binary-search}
Let $T$ be a string of $n$ symbols. Let $T_i, T_k$ and $T_j$ be three lexicographically ordered suffixes of $T$, i.e. $T_i <_L T_k <_L T_j$, where $<_L$ is the lexicographic comparison. Then,
$$\lcp(T_i,T_k)\geq \lcp(T_i,T_j)$$

\end{corollary}

Using Corollary~\ref{cor:binary-search} one can now use the BIS in order to locate $v$ as follows. Begin with the node $w$
corresponding to the lexicographically smallest suffix in the BIS. From the properties of balanced search trees, this node is a leaf.
Consider the list of suffixes on which the BIS is constructed (in the off-line sense). If this list is traversed from $w$ towards $node(T')$, and
for each node $z$ with corresponding suffix $S$ $\lcp(S,T')$ is computed, then the values will increase until the last node $z'$ for
which $\lcp(S,T')\leq k$ is reached. However, such a traversal can take linear time.

Instead, one may use the BIS in order to find $z'$ in $O(\log n)$ time. It should be noted that the reason node $z'$ is the node
being searched for follows directly from the properties and ordering of binary search trees (basically the lexicographical ordering in the BIS can be
substituted with the ordering defined by Corollary~\ref{cor:binary-search}).

Traverse upwards from $w$ in the BIS until a node $z$ is reached
where $z$ is the last node which is an ancestor of $w$ in the BIS whose corresponding suffix is $S$, such that $\lcp(S,T')\leq k$. This
means that either $z$ is the root of the BIS, or the parent of $z$ in the BIS has suffix $S'$ such that $\lcp(S',T') > k$.

Consider the relationship between $z$ and $z'$. $z'$ cannot be in the subtree of the left child of $z$ in the BIS, as
$\lcp(\textit{suffix}(z),T')\leq \lcp(\textit{suffix}(z'),T')$ and  Corollary~\ref{cor:binary-search}. Also, $z'$ has to be in the
subtree rooted by $z$ in the BIS as $\lcp(\textit{suffix}(parent(z)),T')> \lcp(S',T')>k\geq \lcp(\textit{suffix}(z'),T')$ and, hence,
this follows from Corollary~\ref{cor:binary-search}. Thus $z'$ is either $z$, or in the subtree of the right child of $z$ in the BIS.

Begin a recursive traversal down the BIS starting from $z$ where at each node $\hat{z}$ with corresponding suffix $\hat{S}$ do the following.
Compute $\hat{k}=\lcp(\hat{S},T')$. If $\hat{k}>k$ then $z'$ has been passed in the list of suffixes, and the traversal moves into the left
subtree of $\hat{z}$. If $k'=k$ then $z'$ has been found (it does not matter if there are other nodes for which the $\lcp$ is also exactly $k$
as this is sufficient in order to find the entry point). If $\hat{k}<k$ then mark $z'$ as the current candidate, and continue to traverse down
the right subtree of $z'$. If a leaf is reached, then $z'$ is the last candidate that has been marked.

This gives us the following:

\begin{lemma}
It is possible to find the entry point in $O(\log n)$ time.
\end{lemma}

\begin{proof} Once $z'$ is found, the $lca$ of $z'$ and $\mathit{node}(T')$ can be located in constant time,
which as explained above suffices for finding the entry point. The traversal on the BIS that was used in order to find $z'$ takes $O(\log n)$ time
as a simple traversal is used up and down the BIS, spending constant time at each node traversed.
\end{proof}

\subsection{Text Links}
As noted, many applications of the suffix tree use various pointers to the text in order to save
space for labeling the edges. To implement these applications one utilizes the fact that each edge label is a substring of the text. Specifically one can maintain two pointers per edge, one to the location of the first character in the substring in the text, and one pointer to the location of the last character of the substring. Note that if the text contains more than one appearance of this substring, one may pick an arbitrary appearance. Such pointers are called text links, and can still be maintained in the on-line scenario as follows.

When a new leaf $u$ is inserted as a new child of a node $v$ together with the new edge $e=(v,u)$, two text links are created to denote the substring in the text corresponding to $e$. To do this we note that at this time the suffix corresponding to $u$ is the text itself, and hence $label(u,v)$ is simply the last $n-|label(v)|+1$ characters in the text. Thus, the first text link is to location $|label(v)|$ in the text, and the second text link is to the last location of the text.

When breaking an edge $e=(v,u)$ into two by adding a new node $w$, creating edges $e_1=(v,w)$ and $e_2=(w,u)$, we note that the first text link of $e_1$ is the same as the first text link of $e$, and the second text link of $e_2$ is the same as the second text link of $e$. The second text link of $e_1$ is to the location which is $|label(w)|-|label(v)|-1$ away from the first text link, and the first text link of $e_2$ is to the location $|label(w)|-|label(v)|$.

\subsection{Deletions}
\label{sub:deletions}

Assume we built the suffix tree for the string $aT$ where $a$
is a character and $T$ is a text of size $n$. We now wish to support
deletion the first character $a$, hence removing the suffix $aT$
from the suffix tree. This is done by removing $\mathit{node}(aT)$ and
possibly its parent from the suffix tree, and also removing
$\mathit{node}(aT)$ from the BIS. Clearly, this can all be done in $O(\log n)$ time.

Note that we chose each text link to point to the substring in the text which created the edge originally. Therefore,
even if we added many nodes that broke edges within the original edge and then deleted them, we will still always have the appropriate substring in the correct location.
This is because of the stack-like behavior of adding and removing characters to or from the beginning of the text.

Finally we conclude the following:
\begin{theorem}
It is possible to construct a suffix tree in the online text
scenario where the cost of an addition of a character or a deletion
of a character in $O(\log n)$ worst case time, where $n$ is the size
of the text seen so far. Furthermore, at any point in time, an
indexing query can be answered in time $O(m+\log |\Sigma|+\mathit{tocc})$
where $m$ is the size of the pattern, $\Sigma$ is the alphabet
consisting only of characters seen in the text, and $\mathit{tocc}$ is the
number of occurrences of the pattern in the text.

\end{theorem}

\bibliographystyle{plain}
\bibliography{string}

\end{document}